\documentclass[12pt]{article}
\usepackage{amsfonts,amssymb,amsmath}

\newtheorem{theorem}{Theorem}[section]
\newtheorem{lemma}[theorem]{Lemma}
\newtheorem{proposition}[theorem]{Proposition}

\newtheorem{defin}[theorem]{Definition}

\newenvironment{proof}{\noindent \textbf{Proof: }}{\hfill
$\Box$  \vspace{1ex}}
\newenvironment{definition}{\begin{defin}\em}{\end{defin}}
\newtheorem{defins}[theorem]{Definitions}

\newtheorem{exs}[theorem]{Examples}

\newtheorem{ex}[theorem]{Example}
\newenvironment{example}{\begin{ex}\em}{\end{ex}}
\newtheorem{rem}[theorem]{Remark}

\newtheorem{rems}[theorem]{Remarks}

\newtheorem{corollary}[theorem]{Corollary}

\def\fq{\mathbb{F}_q}

\def\B{\mathbf{B}}
\def\lm{\textrm{LM}}

\def\C{\mathcal{C}}
\def\fqt{\mathbb{F}_q[\boldsymbol{t}]}
\def\L{\mathcal{L}}
\def\M{\mathcal{M}}
\def\e{\boldsymbol{e}}
\DeclareMathOperator*{\Supp}{Supp}
\def\a{\alpha}
\def\b{\beta}
\def\r{\beta}

\begin{document}

\begin{center}
{\Large\textbf{On the number of minimal and next-to-minimal weight codewords of 
toric 
codes over hypersimplices}}
\end{center}
\vspace{3ex}

\noindent\begin{center} 
\textsc{C\'{\i}cero  Carvalho\footnote[1]{Instituto de Matem\'atica e 
Estat\'{\i}stica - UFU - Brazil -- cicero@ufu.br} and Nupur 
Patanker\footnote[2]{Indian Institute of Science Bangalore - India -- 
nupurp@iisc.ac.in} }
%\textsc{C\'{\i}cero  Carvalho andNupur Patanker}\footnote{Both authors 
%were partially supported by  grants from CNPq and FAPEMIG.} \\
%{\em\small Faculdade de Matem\'atica \\ Universidade Federal de Uberl\^andia 
%\\  Av.\ J.\ N.\ \'Avila 2121, 38.408-902 - Uberl\^andia - MG, Brazil 
%%\\ Tel/Fax: +55-34-3239-4156 
%\\{\small  cicero@ufu.br \ \ \ \ \ \ \ \ \  gonzalo@famat.ufu.br}} 
\end{center}
\vspace{3ex}

\vspace{4ex}
\noindent
\textbf{Abstract.}
Toric codes are a type of evaluation code introduced by J.P.\ Hansen in 2000. 
They are produced by evaluating (a vector space composed by) polynomials at the 
points of 
$(\fq^*)^s$,  the monomials of these polynomials being related to a certain 
polytope. Toric codes related to hypersimplices are the result of the 
evaluation of a vector space of homogeneous monomially square-free  polynomials 
of degree 
$d$.
The dimension and minimum distance of toric codes related to hypersimplices 
have been determined by 
Jaramillo et al.\ in 2021. 
The next-to-minimal weight in the case $d = 1$ has been determined by 
Jaramillo-Velez et al.\ in 2023, and has been determined
in the cases where $3 \leq d \leq 
\frac{s - 2}{2}$ or $\frac{s + 2}{2} \leq d < s$, by 
Carvalho and Patanker in 2024.
In this work we characterize and determine the number of minimal (respectively, 
next-to-minimal) weight codewords 
when $3 \leq d < s$ (respectively,  
 $3 \leq d \leq 
\frac{s - 2}{2}$ or $\frac{s + 2}{2} \leq d < s$).
\vspace{3ex}


\vspace{1ex}

\vspace{3ex}
\section{Introduction}

Let $\fq$ be a finite field with $q$ elements. In this work we study the class 
of toric codes, introduced by J.P.\ Hansen in 2000 (see \cite{hansen}). These 
codes may be seen as elements in the class of the so-called evaluation  codes 
(see e.g.\ the Introduction of \cite{little}). We study a special case of toric 
codes, which we describe now.

Let $X := (\mathbb{F}_q^*)^s$, then the ideal of all polynomials in 
$\fqt := \fq[t_1, \ldots, t_s]$ which vanish on all points of $X$ is
  $I_X = (t_1^{q - 1} - 
1,\ldots, t_s^{q - 1} - 1)$. It is not difficult to check that, 
writing $n := | X |$ and $X := \{P_1, \ldots, P_n\}$,  the 
evaluation map
\begin{equation}\label{varphi}
\begin{array}{rcl}
\varphi : \fqt/I_X & \longrightarrow & \fq^n \\
f + I_X & \longmapsto & (f(P_1), \ldots, f(P_n))
\end{array}
\end{equation}
is an isomorphism (see e.g. \cite[Prop. 3.7]{gb-in-coding}).

\begin{definition}
Let $d$ be a positive integer such that $d \leq s$,  
%and let $R(d)$ be the 
%convex hull of the set 
%$\{\e_{i_1} + \cdots + \e_{i_d} \mid 1 \leq i_1 < \cdots < i_d \leq s \}$, 
%where $\e_i$ denotes the $i$-th vector in the canonical basis for 
%$\mathbb{R}^s$, $1 \leq i \leq s$. 
let $\L(d) \subset \fqt/I_X$ be the 
$\fq$-vector subspace generated by 
\[
\{ t_1^{a_1} \cdots t_s^{a_s} + I_X 
\mid a_i \in \{0,1\} \; \forall \; i = 1, \ldots, s \textrm{ and } \sum_i a_i = 
d \}.
\]
The toric 
code $\C(d)$ is the image $\varphi(\L(d))$.
\end{definition}

The connection of the above definition with that of \cite{hansen} is that,
denoting by $\Delta_{s,d}$ the $(s,d)$-hypersimplex in $\mathbb{R}^s$, 
i.e. the convex polytope generated by
the set 
$\{\e_{i_1} + \cdots + \e_{i_d} \mid 1 \leq i_1 < \cdots < i_d \leq s \}$, 
where $\e_i$ denotes the $i$-th vector in the canonical basis for 
$\mathbb{R}^s$, $1 \leq i \leq s$, then 
\[
\Delta_{s,d} \cap \mathbb{Z}^s = \{\e_{i_1} + \cdots + \e_{i_d} \mid 1 \leq i_1 < \cdots < i_d \leq s \}
\]
and the $s$-tuples in $\Delta_{s,d} \cap \mathbb{Z}^s$ correspond to the exponents of the monomials in the generating set
$\L(d)$.

The minimum 
distance of $\C(d)$ was determined in \cite[Thm. 4.5]{evalcodes}, and for $q 
\geq 3$ and $d \geq 1$ is as follows:

\[
\delta(\C(d)) = \left\{ \begin{array}{l} (q - 2)^d (q - 1)^{s - d} \textrm{ if 
} 1 \leq d \leq \frac{s}{2} \; ;\\[0.3cm]
(q - 2)^{s -d} (q - 1)^{d} \textrm{ if } 
\frac{s}{2} < d < s .
                           \end{array} 
                 \right.        
\]

In what follows we will always assume that $q \geq 4$ and $d \geq 3$, as in 
\cite{car-nupur}.
The second least Hamming weight of $\C(d)$,  
also known as next-to-minimal 
weight, is denoted by $\delta_2(\C(d))$ and 
was determined in \cite{jaramillo2023} for $d = 1$, and in \cite{car-nupur} for 
$d$ such that $3 \leq d \leq 
\frac{s - 2}{2}$ or $\frac{s + 2}{2} \leq d < s$, see \cite[Thm. 4.5 and Corol. 
4.6]{car-nupur}:
\[
\delta_2(\C(d)) = \left\{ \begin{array}{l} (q - 2)^d (q - 1)^{s - d} + (q - 2)^d(q - 1)^{s - d - 2}  \textrm{ if 
} 3 \leq d \leq \frac{s - 2}{2} \; ;\\[0.3cm]
(q - 2)^{s - d} (q - 1)^{d} + (q - 2)^{s - d}(q - 1)^{d - 2}  \textrm{ if } 
\frac{s+2}{2} < d < s .
                           \end{array} 
                 \right.        
\]

Let $A_i$ be the number of codewords of $\C(d)$ of weight $i$, for $i = 0,\ldots, n$. The weight enumerator polynomial of 
$\C(d)$ is $W_{\C(d)}(X,Y) = \sum_{i = 0}^n A_i X^{n - i}Y^i$. This polynomial is important to determine the
probability of error in error-detection (see e.g. \cite{torleiv}). Clearly, $A_0 = 1$. In this paper we determine
the number of minimal weight codewords (see Theorem \ref{num-min-words} and 
Corollary 
\ref{corol.2.4}) and also the number of
next-to-minimal weight codewords (see Theorem  \ref{num-ntm-codew} and  
Corollary  \ref{num-ntm-codew2}), 
which are the first two values of $A_i$, with $i > 0$, which are nonzero.   

We also characterize the classes of polynomials in $\L(d)$ whose evaluation produces 
minimum weight codewords (see Theorem \ref{min-word-general}) and those whose evaluation produces
next-to-minimal weight codewords (see Theorem \ref{second-min-word-general}). 
These results are used to count the number of codewords mentioned above,
but also have geometric interpretations. For example, from Theorem \ref{min-word-general} one may deduce that any
hypersurface of degree $d$ in $\mathbb{F}_q^s$, given by a homogeneous polynomial in 
$\fqt$ whose monomials are square-free,
and which intersects the affine torus $(\mathbb{F}_q^*)^s$ 
in the maximal number of points (maximal when considered only hypersurfaces of 
this type) must be a specific hyperplane configuration, as described in 
the statements of Theorem \ref{min-word-general} and Corollary \ref{corol.2.4}. 
A similar statement applies for the second maximal number of points 
in the intersection of the affine torus and hypersurfaces of this type. 
We prove that if $2d + 2 \leq s$ then the second maximal number of points 
is attained only if the hypersurface is a certain hyperplane arrangement
(see Theorem \ref{second-min-word-general}), while if $2 d - 2 \geq s$
then the hypersurface may not be a hyperplane arrangement (see Example 
\ref{last-ex}), a phenomenon which also occurs when we look for the 
next-to-minimal weights of projective Reed-Muller codes (see \cite[Prop.\ 
3.3]{car-neu-2018}).

In this paper we work frequently with polynomials in $\fqt$ whose monomials are 
not multiple of $t_i^2$ for all $i = 1, \ldots, s$. We call these polynomials 
monomially square-free, following \cite{evalcodes} (note that in other works, 
e.g.\ \cite{car-nupur}, they are called square-free 
polynomials).

The paper is organized as follows: the next section presents the results 
related to minimal weight codewords, while the last 
section presents the results related to the next-to-minimal weight codewords.

\section{Characterization and number of minimum weight codewords} \label{sec2}

In \cite{car-nupur} the next-to-minimal weights of 
$\C(d)$ were determined, using techniques involving results from Gr\"{o}bner 
basis theory,  for the cases when $3 \leq d \leq \frac{s - 2}{2}$ or 
$\frac{s+2}{2} \leq d < s$. To do that, given 
a homogeneous monomially square-free polynomial $f \in \fqt$ of degree $d$, we 
assumed, after a relabeling of the variables, and after choosing the 
graded-lexicographic order $\prec$ in $\fqt$ with $t_s \prec \cdots \prec t_1$, 
 that the leading monomial of $f$ is $\lm(f) = t_1. \cdots . 
t_d$, and we determined the two lowest possible values for the weight of $f$, 
the lowest being, of course, the minimum distance, already determined in 
\cite{evalcodes}. 
Among other results, we proved the following.

\begin{proposition}\cite[Prop.\ 3.3]{car-nupur} \label{min-word}
Let $f \in \fqt$ be a homogeneous, monic, monomially square-free   
polynomial of degree $d$, such that $\lm(f) = t_1 . \cdots. t_d$, and assume 
that 
$2d \leq s$. 
Then $\varphi(f + I_X)$ is a minimum weight codeword if and only if $f = (t_1 + 
\a_{1} t_{c_1}). \cdots . (t_d + \a_{d} t_{c_d})$, with $c_1, \ldots, c_d \in 
\{d+1, \ldots, s\}$ and  $\a_{1}, \ldots, \a_{d} \in \fq^*$.
\end{proposition}

Now we want to describe all possible 
homogeneous monomially square-free   
polynomials $f$ of degree $d$ such that $\varphi(f + I_X)$ is a minimum weight 
codeword, and for that we examine more closely the relabeling of variables 
mentioned above.

We start with a (nonzero) homogeneous monomially square-free   
polynomial $f$ of degree $d$ in $\fqt$, not endowed with a monomial order, at 
the moment. Let $t_{i_1} . \cdots . t_{i_d}$ be a monomial of $f$. Let $\sigma$ 
be a permutation of $\{1, \ldots, s\}$ such that $\sigma(i_\ell) = \ell$ for 
$\ell = 1, 
\ldots, d$. We will also denote by $\sigma$ the isomorphism $\sigma: \fqt 
\rightarrow \fqt$ defined by $\sigma(\sum \a_M M ) = \sum \a_M \sigma(M)$, 
where 
$\a_M \in \fq$ and if $M = t_{j_1} . \cdots . t_{j_d}$ then $\sigma(M) = 
t_{\sigma(j_1)} . \cdots . t_{\sigma(j_d)}$. Note that now $t_1. \cdots . t_d$ 
is a monomial of $\sigma(f)$ and that this isomorphism is also 
an isomorphism when restricted to  
the $\fq$-vector space $S_d$ formed by 
homogeneous monomially square-free polynomials of degree $d$, together with the 
zero 
polynomial. 
 Also, for all  $P = (\r_1, \ldots, 
\r_s) 
\in (\fq^*)^s$ we define $\sigma(P) = (\r_{\sigma(1)}, \ldots, 
\r_{\sigma(s)})$. 
Recalling that $X = \{P_1, \ldots, P_n\}$, 
one may easily check that for any monomial $M \in \fqt$ we get  
\[
(M(P_1), \ldots, M(P_n)) = (\sigma(M)(\sigma^{-1}(P_1)),\ldots, 
\sigma(M)(\sigma^{-1}(P_n)) ).
\]
Hence 
\begin{equation} \label{eq-f}
(g(P_1), \ldots, g(P_n)) = (\sigma(g)(\sigma^{-1}(P_1)),\ldots, 
\sigma(g)(\sigma^{-1}(P_n)) ).
\end{equation}
for all $g \in S_d$. 
A consequence of this is that the code $\C(d)$, which is obtained by evaluating 
all  
$g \in S_d$ at the sequence of points $(P_1, \ldots, P_n)$ is the same code we 
obtain when we evaluate all $\sigma(g) \in S_d$, with $g \in S_d$, at the 
sequence 
of points 
$(\sigma^{-1}(P_1), \ldots, \sigma^{-1}(P_n))$. Thus, the code 
$\widetilde{C}(d)$ obtained
by evaluating all  $\sigma(g) \in S_d$, with $g \in S_d$, at the sequence of 
points $(P_1, \ldots, P_n)$ is monomially equivalent to $\C(d)$. 

Because of equation \eqref{eq-f}, if we want to study the weight of $\varphi(f 
+ I_X)$ we may, equivalently, study the weight of $\varphi(\sigma(f) + I_X) \in 
\widetilde{C}(d)$. Now we endow $\fqt$ with the graded lexicographic order 
where $t_s \prec \cdots \prec t_1$, so that $\lm(\sigma(f)) = t_1. \cdots . 
t_d$.  We also may assume that $\sigma(f)$ is monic. In the paper 
\cite{car-nupur}, instead of working with $\sigma(f)$ and the code 
$\widetilde{C}(d)$, we worked with $\C(d)$ and wrote that, after a relabeling 
of 
the 
variables, we
may assume that $\lm(f) = t_1. \cdots . t_d$, and then proved 
Proposition \ref{min-word}. After the above considerations, we see that this 
proposition states that if $\sigma(f)$ is a homogeneous, monic, monomially 
square-free   
polynomial of degree $d$, such that $\lm(\sigma(f)) = t_1 . \cdots. t_d$, then 
$\varphi(\sigma(f) + I_X)$ is a minimum weight codeword of $\widetilde{C}(d)$ 
if 
and only if  
$\sigma(f) = (t_1 + 
\a_{1} t_{c_1}). \cdots . (t_d + \a_{d} t_{c_d})$, with $c_1, \ldots, c_d \in 
\{d+1, \ldots, s\}$ and  $\a_{1}, \ldots, \a_{d} \in \fq^*$.

\begin{theorem} \label{min-word-general}
Let $f \in \fqt$ be a homogeneous monomially square-free   
polynomial of degree $d$, and assume 
that 
$2d \leq s$. Then 
$\varphi(f + I_X)$ is a minimum weight codeword of $\C(d)$ if and only if $f$ 
may 
be (uniquely) written as $f = \a 
(t_{b_1} + 
\a_{1} t_{c_1}). \cdots . (t_{b_d} + \a_{d} t_{c_d})$, with $\a, \a_{1}, 
\ldots, \a_{d} \in \fq^*$,  $b_1,\ldots, b_d, c_1, 
\ldots, c_d$ are $2 d $ distinct elements of $\{1,\ldots, s\}$, 
$b_i < c_i$ for all $i = 1, \ldots, d$ and $b_1 < \cdots < b_d$.
\end{theorem}
\begin{proof}
For any $f \in S_d$ such that $\varphi(f + I_X)$ is a minimum weight 
codeword of $\C(d)$ we may find a permutation $\sigma$ such that 
$\lm(\sigma(f)) 
= t_1. \cdots . t_d$, and clearly 
$\varphi(\sigma(f) + I_X)$ is a minimum weight codeword of $\widetilde{C}(d)$. 
Thus, for some $a \in \fq^*$ the polynomial $a \sigma(f)$ is monic, and   
Proposition \ref{min-word} describes the form of $a \sigma(f)$. Applying 
the isomorphism $\sigma^{-1} : S_d \rightarrow S_d$ to $\sigma(f)$ we get that 
$f = a 
(t_{b_1} + \a_{1} t_{c_1}). \cdots . (t_{b_d} + \a_{d} t_{c_d})$, where $a, 
\a_{1}, 
\ldots, \a_{d} \in \fq^*$  and 
$b_1,\ldots, b_d, c_1, 
\ldots, c_d$ are $2 d $ distinct elements of $\{1,\ldots, s\}$. 
To obtain a unique description for each 
polynomial, we observe that since 
 $t_{b_i} + \a_{i} t_{c_i} = \a_{i}(t_{c_i} 
+ \a_{i}^{-1}
t_{b_i})$ for any $i \in \{1, \ldots, d\}$, we may assume that 
$b_i < c_i$ for all $i = 1, \ldots, d$, and after a relabeling 
of the $b_i$'s (and the corresponding $c_i$'s) we may also assume that
$b_1 < \cdots < b_d$.
\end{proof}

Observe that if $f$ is as in the statement of the above proposition, then 
$\lm(f) =t_{b_1} . \cdots . t_{b_d}$. For the proof of the next result, we 
recall that the support of a codeword $\boldsymbol{v}$, denoted by 
$\Supp(\boldsymbol{v})$, is the set of points $P \in (\fq^*)^s$ corresponding 
to positions where $\boldsymbol{v}$ has nonzero entries.

\begin{theorem} \label{num-min-words}
The number of minimal weight codewords of $\C(d)$, in the case where $2 d \leq 
s$ is 
\[
\frac{(q - 1)^{d+1} \prod_{i = 0}^{2d - 1} (s - i)  }{d! \, 2^d} .
\]
\end{theorem}
\begin{proof}
We start by noting that if $\varphi(f + I_X)$ is a minimum weight codeword, 
where  
$f =  
(t_{b_1} + \a_{1} t_{c_1}). \cdots . (t_{b_d} + \a_{d} t_{c_d})$, 
then the set $\{ \varphi(a f + I_X) \mid a \in \fq^*\}$ contains  
$q - 1$ distinct minimum weight codewords, so we will consider from now on only 
monic 
polynomials in $S_d$ whose evaluation produces minimum weight codewords. 
The polynomial $f$ is characterized by the triple of $d$-tuples
\[
\big( (b_1, \ldots, b_d), (c_1, \ldots, c_d), (\a_1, \ldots, \a_d)\big),
\]
where $b_i$, $c_i$ and $\a_i$, for all $i = 1, \ldots, d$, are as in the 
statement of Theorem \ref{min-word-general}. 
We check if polynomials corresponding to distinct triples may produce the same 
codeword. Let 
\[
\big( (b'_1, \ldots, b'_d), (c'_1, \ldots, c'_d), (\a'_1, \ldots, \a'_d)\big)
\]
be the triple to which corresponds the polynomial $g$. Suppose that
$(b_1, \ldots, b_d) \neq (b'_1, \ldots, b'_d)$, then there 
exists $j \in \{1, \ldots, d\}$ such that $b_i = b'_i$  if $i < j$, and 
$b_j \neq b'_j$, and we assume w.l.o.g.\ that $b_j < b'_j$.  Let $P = (\beta_1, 
\ldots, \beta_s) \in (\fq^*)^s$ be such that $\beta_{b'_i} \neq - \a'_i 
\beta_{c'_i}$ for all $i = 1, \ldots, d$
so that $P \in  \Supp(\varphi(g + I_X))$, and such that 
$\beta_{b_j} = - \a_j \beta_{c_j}$, so that $P \notin  \Supp(\varphi(f + 
I_X))$. Thus $\varphi(f + I_X)  \neq \varphi(g + I_X)$, and we assume from now 
on that $(b_1, \ldots, b_d) = (b'_1, \ldots, b'_d)$.
Suppose that $(c_1, \ldots, c_d) \neq (c'_1, \ldots, c'_d)$, 
then there 
exists $j \in \{1, \ldots, d\}$ such that $c_i = c'_i$  if $i < j$, and 
$c_j \neq c'_j$, and we assume w.l.o.g.\ that $c_j < c'_j$.  
Let $P = (\beta_1, 
\ldots, \beta_s) \in (\fq^*)^s$ be such that $\a_j 
\beta_{c_j} = - \beta_{b_j}$, so that $P \notin \Supp(\varphi(f + I_X))$, and 
also such that  
$\a'_i  \beta_{c'_i} \neq - \beta_{b_i}$ for all $i = 1,\ldots, d$, so that 
$P \in \Supp(\varphi(g + I_X))$. Again we have 
$\varphi(f + I_X)  \neq \varphi(g + I_X)$, and we assume furthermore from now 
on that $(c_1, \ldots, c_d) = (c'_1, \ldots, c'_d)$. Suppose that 
$(\a_1, \ldots, \a_d) \neq (\a'_1, \ldots, \a'_d)$, and let $j \in \{1, \ldots, 
d\}$ be such that $\a_j \neq \a'_j$. Let $P = (\beta_1, \ldots, \beta_s) 
\in (\fq^*)^s$ be such that $\beta_{b_j} = -\a_j \beta_{c_j}$, so that $P 
\notin 
\Supp(\varphi(f + I_X))$ and $\beta_{b_j} \neq - \a'_j \beta_{c_j}$, and also 
such that $\beta_{b_i} \neq - \a'_i \beta_{c_i}$ for all $i \in \{1, \ldots, 
d\} 
\setminus \{j\}$. Then $P \in \Supp(\varphi(g + I_X))$ and 
$\varphi(f + I_X)  \neq \varphi(g + I_X)$. 

This completes the proof that 
each polynomial of the form $f =  
(t_{b_1} + \a_{1} t_{c_1}). \cdots . (t_{b_d} + \a_{d} t_{c_d})$, 
where  $\a_1, 
\ldots, \a_d \in \fq^*$,  $b_1,\ldots, b_d, c_1, 
\ldots, c_d$ are $2 d $ distinct elements of $\{1,\ldots, s\}$, 
$b_i < c_i$ for all $i = 1, \ldots, d$ and 
$b_1 < \cdots < b_d$
produces a 
distinct minimum weight codeword.
We want to count the number of such polynomials. We start by choosing  pairs 
$(b_i, c_i)$, with $1 \leq b_i < c_i \leq s$, and $i = 1, \ldots, d$. To choose 
the first pair we have $\binom{s}{2}$ possibilities (the least number of the 
pair will be $b_i$). For the second pair we have $\binom{s - 2}{2}$ 
possibilities, and so on. After choosing $d$ pairs we may order them in 
increasing order of the first entry, to get the sequence $((b_1, c_1), \ldots, 
(b_d, c_d))$. Note that there are $d!$ ways of arriving at the same sequence 
using this process. Thus we have 
\[
\frac{1}{d!} \; \prod_{k = 0}^{d - 1} \binom{s - 2 k }{2} = 
 \frac{s(s-1) . \cdots . (s - 2d + 2)(s - 2d + 1)}{d! \,2^d} 
\]
possibilities for distinct sequences 
$((b_1, c_1), \ldots, (b_d, c_d))$, where $b_i < c_i$ for all $i = 1, \ldots, 
d$ and $b_1 < \ldots, < b_d$. For the $d$-tuple $(\a_1, \ldots , \a_d)$ we have 
$(q - 1)^d$ possibilities. Thus we have a total of 
\[
(q - 1)^d \frac{\prod_{i = 0}^{2d - 1} (s - i)}{d! \, 2^d}
\]
monic polynomials of the form $f =  
(t_{b_1} + \a_{1} t_{c_1}). \cdots . (t_{b_d} + \a_{d} t_{c_d})$, 
where  $\a_{1}, 
\ldots, \a_{d} \in \fq^*$,  $b_1,\ldots, b_d, c_1, 
\ldots, c_d$ are $2 d $ distinct elements of $\{1,\ldots, s\}$, 
$b_1 < \cdots < b_d$ and $b_i < c_i$ for all $i = 1, \ldots, d$. Finally, from 
what we have done above,  we get that there are exactly 
\[
\frac{(q - 1)^{d+1} \prod_{i = 0}^{2d - 1} (s - i)  }{d! \, 2^d} 
\]
codewords of minimum weight in $\C(d)$.
\end{proof}

To characterize the number of minimal weight codewords and find their number, for $d$ in the range 
$s < 2d < 2s$ we use a distinctive characteristic of toric codes defined over hypersimplex, namely
that $\C(d)$ and $\C(s - d)$ are monomially equivalent. This equivalence is a consequence of two facts:
first, the bijection between 
\[
L(d) := \{ t_1^{a_1} \cdots t_s^{a_s}  
\mid a_i \in \{0,1\} \; \forall \; i = 1, \ldots, s \textrm{ and } \sum_i a_i = 
d \}
\]
and
\[
L(s - d) = \{ t_1^{a_1} \cdots t_s^{a_s}  
\mid a_i \in \{0,1\} \; \forall \; i = 1, \ldots, s \textrm{ and } \sum_i a_i = 
s - d \}
\]
given by $M := t_1^{a_1} \cdots t_s^{a_s} \mapsto M^c :=  t_1^{1 -a_1} \cdots t_s^{1 -a_s}$, and second, 
the bijection between the points of $X$ given by 
$P_i := (\b_{i 1}, \ldots, \b_{i s}) \mapsto  Q_i :=(\b_{i 1}^{-1}, \ldots, 
\b_{i s}^{-1})$, 
for all $i = 1,\ldots, n$. Clearly $\{P_1, \ldots, P_n\} = \{Q_1,\ldots, Q_n\}$ and for any $M \in L(d)$ we get
$M(P_i) = (\prod_{j = 1}^s \b_{i j}) M^c(Q_i)$ for all $P_i \in X$. The
bijection between $L(d)$ and $L(s - d)$ may be extended to the vector space 
they generate, so that if $f$ is a homogeneous monomially square-free  
polynomial of degree 
$d$ we have 
\begin{equation} \label{f-fc}
f(P_i) = (\prod_{j = 1}^s \b_{i j}) f^c(Q_i)
\end{equation}
for all $P_i \in X$.

From this it is easy 
to deduce that
we may obtain $\varphi(\L(d))$ from $\varphi(\L(s - d))$  after a reordering of the $s$-tuples of 
$\varphi(\L(s - d))$ together with multiplying the entry corresponding to point the $Q_i$ by 
$\prod_{j = 1}^s \b_{i j}$, where $Q_i = (\b_{i 1}^{-1}, \ldots, \b_{i 
s}^{-1})$ for all 
$i = 1, \ldots, n$.

\begin{corollary} \label{corol.2.4}
Assume that
$s < 2d < 2s$ and
let $f \in \fqt$ be a homogeneous, monomially square-free   
polynomial of degree $d$. 
Then
$\varphi(f + I_X)$ is a minimum weight codeword of $\C(d)$ if and only if $f$ 
may 
be (uniquely) written as 
\[
f = \a 
(t_{b_1} + \a_{1} t_{c_1}). \cdots . (t_{b_{s- d}} + \a_{s - d} t_{c_{s - 
d}}) 
\prod_{\stackrel{j = 1}{j \notin A_f}}^s t_j,
\] 
where $\a, \a_{1}, 
\ldots, \a_{s - d} \in \fq^*$, $A_f := \{ b_1,\ldots, b_{s - d}, c_1, 
\ldots, c_{s - d}\} \subset \{1,\ldots, s\}$ 
is a set with  $2 (s - d) $ distinct elements,  
$b_1 < \cdots < b_{s - d}$ and $b_i < c_i$ for all $i = 1, \ldots, {s - d}$.

The number of minimal weight codewords of $\C(d)$ in this case 
is 
\[
\frac{(q - 1)^{s-d+1} \prod_{i = 0}^{2s - 2d - 1} (s - i)  }{(s - d)! \, 2^{s 
-d}} .
\]
\end{corollary}
\begin{proof}
Let $r$ be a positive integer. The bijection $M \mapsto M^c$ defined above 
between square-free monomials of degree $r$ and square-free monomials of degree 
$s -r$ may be extended to a bijection $f \mapsto f^c$ between the spaces of 
polynomials 
generated by these two sets of monomials.

Assume  that $r < s/2$, and let $A_r 
:= 
\{u_1, \ldots, u_r, v_1, \ldots, v_r\} \subset \{1, \ldots, s\}$ be a set of 
$2r$ distinct elements. Let $\gamma_1, \ldots, \gamma_r \in \fq^*$. We claim 
that  
\[
\left( (t_{u_1} + \gamma_1 t_{v_1}). \cdots . (t_{u_r} + \gamma_r t_{v_r}) 
\right)^c = 
(t_{v_1} + \gamma_1 t_{u_1}). \cdots . (t_{v_r} + \gamma_r t_{u_r}) 
\prod_{\stackrel{j = 1}{j \notin A_r}}^s t_j.
\]
We prove the claim by induction. The case $r = 1$ is simple to verify. Assume 
now that $r \geq 2$ and 
that the claim holds for $r - 1$. Then 
\begin{equation*} 
\begin{split}
( (t_{u_1} &+ \gamma_1 t_{v_1}). \cdots . (t_{u_{r-1}} + \gamma_{r - 1} 
t_{v_{r-1}} )
(t_{u_r} + \gamma_r t_{v_r} ) )^c  = \\
 &\left( (t_{u_1} + \gamma_1 t_{v_1}). \cdots . (t_{u_{r - 1}} + \gamma_{r - 1} 
 t_{v_{r- 1}}) t_{u_r} \right)^c \\
 &+ \left( (t_{u_1} + \gamma_1 t_{v_1}). \cdots . (t_{u_{r - 1}} + \gamma_{r - 
 1} 
 t_{v_{r- 1}}) 
  \gamma_r t_{v_r} \right)^c
\end{split}
\end{equation*}
From the definition of the bijection and the induction hypothesis, we get 
\begin{equation*} 
\begin{split}
(t_{u_1} &+ \gamma_1 t_{v_1}). \cdots . (t_{u_{r - 1}} + \gamma_{r - 1} 
 t_{v_{r- 1}}) t_{u_r} )^c = \\ 
 &(t_{v_1} + \gamma_1 t_{u_1}). \cdots . 
 (t_{v_{r - 1}} + \gamma_{r - 1} t_{u_{r - 1}}) 
 \prod_{\stackrel{j = 1}{j \notin A_{r - 1} \cup \{u_r\} }}^s t_j 
\end{split}
\end{equation*}
Similarly 
\begin{equation*} 
\begin{split}
(t_{u_1} &+ \gamma_1 t_{v_1}). \cdots . (t_{u_{r - 1}} + \gamma_{r - 1} 
 t_{v_{r- 1}}) \gamma_r t_{v_r} )^c = \\ 
 &(t_{v_1} + \gamma_1 t_{u_1}). \cdots . 
 (t_{v_{r - 1}} + \gamma_{r - 1} t_{u_{r - 1}}) \gamma_r
 \prod_{\stackrel{j = 1}{j \notin A_{r - 1} \cup \{v_r\} }}^s t_j 
\end{split}
\end{equation*}
so that
\begin{equation*} 
\begin{split}
( (t_{u_1} &+ \gamma_1 t_{v_1}). \cdots . (t_{u_{r-1}} + \gamma_{r - 1} 
t_{v_{r-1}} )
(t_{u_r} + \gamma_r t_{v_r} ) )^c  = \\
  &(t_{v_1} + \gamma_1 t_{u_1}). \cdots . 
  (t_{v_{r - 1}} + \gamma_{r - 1} t_{u_{r - 1}})
  (t_{v_{r}} + \gamma_{r} t_{u_{r}})
  \prod_{\stackrel{j = 1}{j \notin A_{r - 1} \cup \{u_r, v_r\} }}^s t_j 
\end{split}
\end{equation*}
which proves the claim since $A_r  = A_{r - 1} \cup \{u_r, v_r\}$.

Now we apply the claim to prove the statement on the characterization of 
minimal weight codewords. We know that $\C(d)$ is monomially equivalent to 
$\C(s - d)$ and from $s < 2d$ we get $2(s - d) < s$. Thus, from Theorem 
\ref{min-word-general} we know the form of the polynomials $f$ whose evaluation 
produces the minimal weight codewords of $\C(s - d)$, and from Equation 
\eqref{f-fc} we get that the minimal weight codewords of $\C(d)$ are obtained 
from the 
evaluation of the polynomials $f^c$. Applying the claim proved above to the 
polynomials in Theorem \ref{min-word-general} and making the same 
normalizations we did at the end of the proof of that Theorem, 
we arrive at the first statement of the present Theorem.

As for the number of minimal weight codewords of $\C(d)$ in the case
$s < 2d < 2s$, from the isomorphism mentioned above, we know that it is equal 
to the  
number of 
minimal weight codewords of $\C(s - d)$, so we get what we want by replacing 
$d$ 
by $s - d$ in the formula of Theorem \ref{num-min-words}.
\end{proof}

\section{Characterization and number of next-to-minimal weight codewords}

We want to characterize the next-to-minimal codewords of $\C(d)$, and  to count 
them. We start with an auxiliary result which will be useful in the proof of 
the main result.

\begin{lemma} \label{nonzeros-linear}
Let $u$ be an integer such that $1 \leq u \leq s$. The number of $s$-tuples in 
$(\fq^*)^s$ which are not zeros of the polynomial $\a_1 t_1 + \cdots + \a_u 
t_u 
\in \fqt$, where 
$\a_1, \ldots, \a_u \in \fq^*$ is 
\[
D_u =  \left(\frac{(q - 1)^{u+1} + (-1)^u}{q} + (-1)^{u+1}\right)(q - 1)^{s - 
u}.
\]
Moreover, for $k$ such that 
$1 \leq 2k - 1 < 2k < 2 k + 1 < 2 k + 2 \leq s$ we have
\[
D_{2 k - 1} > D_{2 k + 1} > \frac{(q - 1)^{s + 1}}{q} > D_{2 k + 2} > D_{2 k}. 
\]
\end{lemma}
\begin{proof}
We denote by $\widetilde{D}_u$ (respectively, $\widetilde{E}_u$) the number of 
$u$-tuples 
in $(\fq^*)^u$ which are not zeros (respectively, are zeros)  of the polynomial 
$\a_1 t_1 + \ldots + \a_u t_u$, where $1 \leq u \leq s$.

If $u = 1$ we get that all $q - 1$ elements in  $\fq^*$ are not zeros 
of the polynomial $\a_1 t_1$, i.e. $\widetilde{D}_1 = q - 1$ and
$\widetilde{E}_1 = 0$.

%If $d = 2$ we get that there are $q - 1$ pairs in $(\fq^*)^2$ which are zeros 
%of $\a_1 t_1 + \a_2 t_2$, i.e. $\widetilde{E}_2 = q - 1$, so there are  
%$\widetilde{D}_2 = (q-1)^2 - (q - 1) = (q-1)(q - 2)$ pairs in $(\fq^*)^2$ 
%which 
%are not zeros of this polynomial.

For $u > 1$ let $(\b_1, \ldots, \b_u) \in (\fq^*)^u$, we consider two cases:\\
i) if $\a_1 \b_1 +  \cdots + \a_{u - 1} \b_{u - 1} =: \gamma \neq 0$ then there 
are $q -2$ values for $\b_u$ such that $\a_1 \b_1 +  \cdots + \a_{u} \b_{u}  
\neq 0$ (since we must have  $\b_u \neq -\gamma/\a_u$);\\
ii) if $\a_1 \b_1 +  \cdots + \a_{u - 1} \b_{u - 1} = 0$ then we have $q - 1$ 
values for $\b_u$ such that $\a_1 \b_1 +  \cdots + \a_{u} \b_{u}  
\neq 0$.

This shows that, for $u > 1$ we have 
\begin{equation*} 
\begin{split}
\widetilde{D}_u &= \widetilde{D}_{u - 1}(q - 2) + \widetilde{E}_{u - 1}(q - 1) 
\\ &=   \widetilde{D}_{u - 1}(q - 2)  + ( (q-1)^{u - 1} - \widetilde{D}_{u - 
1})(q -1) =
(q - 1)^u - \widetilde{D}_{u - 1}.
\end{split}
\end{equation*}

Applying recursively this equality for $\widetilde{D}_{u - 1}, \ldots , 
\widetilde{D}_{2}$ and using that $\widetilde{D}_{1} = q - 1$
we get 
$\widetilde{D}_u = (q - 1)^u - (q - 1)^{u - 1} + \cdots + (-1)^u(q - 1)^2 
+ (-1)^{u - 1} (q - 1)$. 
Then we use that, when $u$ is even we have 
$x^u - x^{u -1} + \cdots + x^2 - x + 1 = \frac{x^{u + 1} + 1}{x + 1}$
and when $u$ is odd we have $x^u - x^{u -1} + \cdots - x^2 + x - 1 = \frac{x^{u 
+ 1} - 1}{x + 1}$, 
so that, 
for $u$ even, say $u = 2k$,  we have
\[
\widetilde{D}_{2k} = \frac{(q - 1)^{2 k + 1} + 1}{q} - 1
\]
while if $u = 2k - 1$ then 
\[
\widetilde{D}_{2k-1} = \frac{(q - 1)^{2 k} - 1}{q} + 1.
\]

Thus, the number of $s$-tuples in 
$(\fq^*)^s$ which are not zeros of the polynomial $\a_1 t_1 + \cdots + \a_u t_u 
\in \fqt$ is equal to
\[
D_{2k} = \left(\frac{(q - 1)^{2 k + 1} + 1}{q} - 1\right)(q - 1)^{s - 2k} 
\]
when $u = 2k$, and when $u = 2k - 1$ is equal to 
\[
D_{2k-1} = \left(\frac{(q - 1)^{2 k} - 1}{q} + 1\right)(q - 1)^{s - 2k + 1}, 
\]

which we subsume by writing 
\[
D_{u} = \left(\frac{(q - 1)^{u+1} + (-1)^u}{q} + (-1)^{u+1}\right)(q - 1)^{s - 
u}.
\]
Let $k$ be such that 
$1 \leq 2k - 1 < 2k < 2 k + 1 < 2 k + 2 \leq s$, from
\begin{equation*} 
\begin{split}
D_{2k}  &= 
\frac{(q-1)^{s+1}}{q} - (q - 1)^{s - 2k}(1 - \frac{1}{q}) \\
D_{2k-1}  &= 
\frac{(q-1)^{s+1}}{q} + (q - 1)^{s - 2k + 1}(1 - \frac{1}{q})
\end{split}
\end{equation*}
we get that 
\[
D_{2 k - 1} > D_{2 k + 1} > \frac{(q - 1)^{s + 1}}{q} > D_{2 k + 2} > D_{2 k}. 
\]
\end{proof}

As mentioned in the beginning of Section 2,  
next-to-minimal weights 
of $\C(d)$ were obtained in \cite{car-nupur} through methods which involved 
results from
Gr\"{o}bner basis theory. In the proof of the next theorem we will need some 
of these results. We 
recall now a concept which plays an 
important role in these methods.
Let $\M$ be the set of monomials in the ring $\fqt$
and let $I \subset \fqt$ be an ideal. The footprint 
of $I$ is the set
\[
\Delta(I) := \{ M \in \M \mid M \neq \lm(f) \textrm{ for all } f \in I, f \neq 
0\}.
\]
If the footprint is finite, then the number of $s$-tuples which are zeros of 
all polynomials in $I$ is at most $| \Delta(I) |$ (see \cite[Thm. 
8.32]{becker}). A 
consequence 
of this is that the weight of $\varphi(f + I_X)$, where $\varphi$ is the 
evaluation map of Equation \eqref{varphi}, is at least $|\Delta(I_X)| - 
| \Delta(I_X + (f)) |$ (see \cite[Prop. 2.4]{car-nupur}), and it's easy to 
check that
\[
|\Delta(I_X)| - | \Delta(I_X + (f)) | \geq | \{ M \in \Delta(I_X) \mid M 
\textrm{ 
is a multiple of } \lm(f) \} |. 
\]
Thus, denoting by $\omega(\varphi(f + I_X))$ the weight of $\varphi(f + I_X)$, 
we get that
\begin{equation} \label{bound-weight}
\omega(\varphi(f + I_X)) \geq | \{ M \in \Delta(I_X) \mid M \textrm{ 
is a multiple of } \lm(f) \} |.
\end{equation}

In what follows, we will use results from Section 3 of \cite{car-nupur}, so 
from now on we assume that $q \geq 4$.

\begin{theorem} \label{second-min-word-general}
Let $f \in \fqt$ be a homogeneous, monomially	 square-free   
polynomial of degree $d$, and assume 
that 
$2d + 2\leq s$. Then 
$\varphi(f + I_X)$ is a next-to-minimal weight codeword of $\C(d)$ if and only 
if $f$ 
may 
be (uniquely) written as 
\[
f = \a \left(\prod_{i = 1}^{d - 1}( t_{b_i} + \a_{i} t_{c_i})\right)
(t_{b_{2d -1}} + \a_{b_{2 d}} t_{b_{2 d}} + \a_{b_{2 d+1}} t_{b_{2d+1}}
+ \a_{b_{2 d+2}} t_{b_{2d+2}}), 
\]
%\marginpar{renomear os ultimos bi's. Corrigir assuncao da ordem gr lex no 
%inicio da sessao 2}
where $\a, 
\a_{1}, \ldots, \a_{d-1}, \a_{b_{2 d}}, \a_{b_{2 d +1}}, \a_{b_{2 d + 2}}  
\in \fq^*$,  $b_1,\ldots, b_{d-1}$, $c_1, \ldots, c_{d-1}$,  $b_{2d - 1}$, 
$b_{2d}$, $b_{2d + 1}$, $b_{2d + 2}$  
 are $2 d + 2 $ distinct elements of $\{1,\ldots, s\}$, 
$b_1 < \cdots < b_d$,  $b_{2d - 1} < b_{2d} < b_{2d + 1} < b_{2d + 2}$ and $b_i 
< 
c_i$ for all $i = 1, \ldots, d - 1$.
\end{theorem}
\begin{proof}
Let $f \in S_d$, $f \neq 0$. As seen in the beginning of Section 
\ref{sec2}, 	 
we may find a permutation $\sigma$ such that, after we endow $\sigma(\fqt)$ 
($=\fqt$) with the graded lexicographic order with $t_s \prec \cdots \prec 
t_1$, we have
$\lm(\sigma(f)) = t_1. \cdots . t_d$. 
We will assume, for the moment, that $\sigma(f)$ is monic.
In the beginning of Section 3 of \cite{car-nupur} it is observed that 
\[
\Delta(I_X) = \left\{ \prod_{i = 1}^s t_i^{a_i} \in \M \mid 0 \leq 
a_i 
\leq q - 2 \; \forall \; i = 1, \ldots, s \right\}.
\]

Thus, from Equation \eqref{bound-weight}
 we get that $\omega(\varphi(\sigma(f) + I_X)) 
\geq (q - 
2)^d (q - 1)^{s - d}$. Assume, from now on, that $2d + 2 \leq s$. In this 
case,  we know that the  
minimum distance of $\widetilde{\C}(d)$ is $(q - 
2)^d (q - 1)^{s - d}$ (see \cite[Thm. 4.5]{evalcodes}), and this 
means that $\varphi(\sigma(f) + I_X)$ is a minimum weight codeword if and only 
if 
$\{t_1^{q -1} - 1, \ldots, t_s^{q - 1} - 1, \sigma(f)\}$ is a Gr\"obner basis 
for $I_X + (\sigma(f))$. 

Assume that $\varphi(\sigma(f) + I_X)$ is a next-to-minimal weight codeword of 
$\widetilde{C}(d)$. 
Since the monomials in $\{t_{d+1}^{q - 1}, \ldots, t_s^{q - 1}, 
\lm(\sigma(f))\}$ are pairwise coprime, we get from \cite[p.\ 103--104]{cox} 
that the set  
$\{t_{d + 1}^{q - 1} - 1, \ldots, t_{s}^{q-1} - 1, \sigma(f)\}$ is a Gr\"obner 
basis for 
the ideal that it defines. We also know that $\{t_1^{q - 1} - 1, \ldots, 
t_s^{q - 1} - 1\}$ is a Gr\"obner basis (for $I_X$). Then, since
$\varphi(\sigma(f) + I_X)$ is not a minimum weight codeword we must have that 
for some $j \in \{1, \ldots, d\}$ the remainder $r_j$ in the division of the  
$S$-polynomial $S(t_j^{q - 1} - 1, \sigma(f))$ by $\{t_1^{q -1} - 1, \ldots, 
t_s^{q - 1} - 1, \sigma(f)\}$  is not zero. In \cite[Thm.\ 3.1]{car-nupur} we 
listed the possibilities for the leading monomial of $r_j$, and if 
$\varphi(\sigma(f) + I_X)$ is a next-to-minimal weight codeword the 
possibilities are (using the notation of \cite{car-nupur}) 
$M_4 := t_1 . \cdots . \widehat{t_j} . \cdots . t_d t_{e_1}^{q - 2} t_{e_2}$
or 
$M_2 :=  t_j^{q - 2} t_{1} . \cdots . \widehat{t_j} . \cdots . \widehat{t_\ell} 
. \cdots t_{d}.t_{e_1}. t_{e_2}$
where a hat over a variable means it does not appear in the product, $\ell \in 
\{1, \ldots, d \} \setminus \{j\}$, and 
$t_{e_1}$ and $t_{e_2}$ are distinct elements in the set $\{t_{d + 1}, 
\ldots, t_s\}$ (see the paragraph just before Theorem 4.5 in \cite{car-nupur}). 
Let $M \in \{M_2, M_4\}$, as a consequence of  \cite[Lemma 4.1]{car-nupur} we 
get that the number of 
monomials which are in $\Delta(I_X)$ and are multiples $M$ but are not 
multiples of $\lm(\sigma(f))$ is equal to $(q - 2)^d(q - 1)^{s - d - 2}$. 
Thus, from  
\cite[Prop. 2.4]{car-nupur} we have that 
\[
\omega(\varphi(\sigma(f) + I_X)) \geq  (q - 2)^d (q - 1)^{s 
- d} + (q - 2)^d(q - 1)^{s - d - 2}.
\]
From \cite[Thm. 4.5]{car-nupur} we 
know that the right side is the value of the next-to-minimal weight of 
$\widetilde{C}(d)$,
and since $\varphi(\sigma(f) + I_X)$ is a next-to-minimal weight codeword of 
$\widetilde{C}(d)$	we must have that 
$\{t_1^{q -1} - 1, \ldots, t_s^{q - 1} - 1, f, r_j\}$ is 
a Gr\"obner basis for the ideal it defines (which is $I_X + (\sigma(f))$). 
In particular, the remainder  in the division of the  
$S$-polynomial $S(t_{j'}^{q - 1} - 1, \sigma(f))$ by $\{t_1^{q -1} - 1, \ldots, 
t_s^{q - 1} - 1, \sigma(f)\}$  is zero, for all $j' \in \{1,\ldots, d\} 
\setminus \{j\}$. Then from \cite[Corol.\ 3.2]{car-nupur} we have 
that for all $j' \in \{1,\ldots, 
d\}\setminus\{j\}$ there exists $\gamma_{j'} \in \fq^*$ and $e_{j'} \in \{1, 
\ldots, s\} 
\setminus \{1, \ldots, d\}$ such that $t_{j'} + \gamma_{j'} t_{e_{j'}} \mid 
\sigma(f)$, hence we must have 
\[
\sigma(f) = \left(\prod_{\substack{i=1 \\ i\neq j}}^d (t_{i} + \gamma_{i} 
t_{e_{i}})\right) f_1
\]
where  $f_1 = t_j$ or $f_1 = t_j + \gamma_{v_2} t_{v_2} + \cdots + \gamma_{v_u} 
t_{v_u}$, with $\gamma_{v_2}, \ldots, \gamma_{v_u} \in \fq^*$ and 
$2 \leq 
u \leq s - 2d + 2$. In case $u \geq 2$ the variables 
$t_{v_2}, \ldots, t_{v_u}$ are distinct,  and also distinct from $t_j$ and all 
the 
variables which appear in $\prod_{\substack{i=1 \\ i\neq j}}^d (t_{i} + 
\gamma_{i} t_{e_{i}})$. For each $i \in \{1, \ldots, d\} \setminus \{ j \}$ the 
number of pairs $(\tau_i, \tau_{e_i}) \in (\fq^*)^2$ such that 
$\tau_{i} + \gamma_{i}  \tau_{e_{i}} \neq 0$ is $(q-1)^2 - (q -1 ) = (q - 2)(q 
- 
1)$, so the number of $(2d - 2)$-tuples 
$(\tau_1, \ldots, \hat{\tau}_j, \ldots, \tau_d, \tau_{e_{1}}, \ldots, 
\hat{\tau}_{e_{j}}, \ldots, \tau_{e_{d}}) \in (\fq^*)^{2d - 2}$ such that 
$\prod_{\substack{i=1 \\ i\neq j}}^d (\tau_{i} + \gamma_{i} 
\tau_{e_{i}}) \neq 0$ is $(q-2)^{d - 1}(q - 1)^{d - 1}$.
To count the number of $(s - 2d + 2)$-tuples which are not zeros of $f_1$, we 
use 
Lemma \ref{nonzeros-linear} (of course, in the statement, we must replace  $s$ 
by $s - 2d + 2$). If $u = 4$, the number of such $(s - 2d + 
2)$-tuples is 
\begin{equation*} 
\begin{split}
D_4 &= \left(\frac{(q - 1)^{5} + 1}{q} -1  \right)(q - 1)^{s - 2d -2} \\ &= 
\left( (q - 2)(q - 1)( (q -1)^2 + 1) \right) (q - 1)^{s - 2d -2}. 
\end{split}
\end{equation*}
Thus the number of $s$-tuples in $(\fq^*)^s$ which are not zeros of $\sigma(f)$ 
in the case $u = 4$ is 
\begin{equation*} 
\begin{split}
(q - &2)^{d - 1}(q - 1)^{d - 1} \left( (q - 2)(q - 1)( (q -1)^2 + 1) \right) (q 
- 
1)^{s - 2d -2} \\ &= (q - 2)^d (q - 1)^{s 
- d} + (q - 2)^d(q - 1)^{s - d - 2}
\end{split}
\end{equation*}
which is the next-to-minimal weight of $\widetilde{\C}(d)$. 

From the inequalities in 
Lemma \ref*{nonzeros-linear} we get that
$D_u \neq D_4$ for all $u \in \{1, \ldots,s - 2 d + 2\}$, $u \neq 4$, so
 $f_1$ must have exactly four 
variables, i.e. $f_1 = t_j + \gamma_{v_2} t_{v_2} + \gamma_{v_3} t_{v_3} + 
\gamma_{v_4} t_{v_4}$, with $\{\gamma_{v_2}, \gamma_{v_3}, \gamma_{v_4}\} 
\subset \fq^*$.

We assumed earlier that $\sigma(f)$ is monic, in the general case we see that 
if $\sigma(f)$ is a next-to-minimal weight codeword of $\widetilde{\C}(d)$ then 
it has the form
\[
\sigma(f) = \gamma \left(\prod_{\substack{i=1 \\ i\neq j}}^d (t_{i} + 
\gamma_{i} 
t_{e_{i}})\right) (t_j + \gamma_{v_2} t_{v_2} + \gamma_{v_3} t_{v_3} + 
\gamma_{v_4} t_{v_4}) 
\]
where $\gamma \in \fq^*$. Thus, applying 
the isomorphism $\sigma^{-1} : S_d \rightarrow S_d$ to $\sigma(f)$
we get that if $f$ is a 
next-to-minimal weight codeword of $\C(d)$ then $f$ can be written as in the 
statement. We have already commented on the uniqueness of such a form in the 
proof 
of 
Theorem \ref{min-word-general}. From the calculations above we get that if $f$ 
has this form then $\omega(\varphi(f + I_X)) = (q - 2)^d (q - 1)^{s 
- d} + (q - 2)^d(q - 1)^{s - d - 2}$, which finishes the proof of the theorem.
\end{proof}

Note that if $f$ is as in the statement of the above theorem, then $\lm(f) = 
t_{b_1}. \ldots t_{b_{d - 1}} . t_{b_{2d -1}}$.
Now we can count the number of next-to-minimal weight codewords in $\C(d)$.

\begin{theorem} \label{num-ntm-codew}
The number of next-to-minimal weight codewords of $\C(d)$, in the case where $2 
d + 2\leq s$ is 
\[
\frac{(q-1)^{d+3}\prod_{i = 0}^{2d + 1} (s - i) }{(d - 1)! \; 2^{d+2} . 3}.
\]
\end{theorem}
\begin{proof}
Let $f$ be as in the statement of Theorem \ref{second-min-word-general}, and 
assume that $f$ is monic (so the set $\{\varphi(\alpha f + I_X) \mid \alpha \in 
\fq^*\}$ contains $q-1$ distinct codewords with the same support). We know that 
$f$ is characterized by the 5-tuple 
\begin{equation*} 
\begin{split}
\big( (b_1, \ldots, b_{d -1}), &(c_1,\ldots, c_{d - 1}),(\alpha_1, \ldots, 
\alpha_{d - 1}), (b_{2d - 1}, b_{2d}, b_{2d + 1}, b_{2d + 2}), \\ &
(\a_{2d}, \a_{2d + 1}, \a_{2d + 2}) \big) \in \mathbb{N}^{d - 1} \times 
\mathbb{N}^{d - 1} \times (\fq^*)^{d - 1} \times \mathbb{N}^{4} \times 
(\fq^*)^3,
\end{split}
\end{equation*}
where $\mathbb{N}$ is the set of positive integers and the $b_i$'s, $c_j$'s and 
$\a_k$'s have the restrictions which appear in the statement of Theorem 
\ref{second-min-word-general}.

Let $g$ be a homogeneous, monic,  monomially square-free polynomial of degree 
$d$ such that 
$\varphi(g + 
I_X)$ is a next-to-minimal weight codeword and let 
\begin{equation*} 
\begin{split}
\big((b'_1, \ldots, b'_{d -1}), &(c'_1,\ldots, c'_{d - 1}),(\alpha'_1, \ldots, 
\alpha'_{d - 1}), (b'_{2d - 1}, b'_{2d}, b'_{2d + 1}, b'_{2d + 2}), \\ &
(\a'_{2d}, \a'_{2d + 1}, \a'_{2d + 2}) \big) \in \mathbb{N}^{d - 1} \times 
\mathbb{N}^{d - 1} \times (\fq^*)^{d - 1} \times \mathbb{N}^{4} \times 
(\fq^*)^3
\end{split}
\end{equation*}
be the 5-tuple associated to $g$. Similarly to what was done in the 
proof of Theorem \ref{num-min-words}, one may show that 
if $(b_1, \ldots, b_{d -1}) \neq (b'_1, \ldots, b'_{d -1})$ or
$(c_1,\ldots, c_{d - 1}) \neq (c'_1,\ldots, c'_{d - 1})$ or 
$(\alpha_1, \ldots, \alpha_{d - 1}) \neq (\alpha'_1, \ldots, \alpha'_{d - 1})$
then there exists $P \in (\fq^*)^s$ such that 
$P \in \Supp(\varphi(g + I_X))$ and $P \notin \Supp(\varphi(f + I_X))$.
So we assume that
$(b_1, \ldots, b_{d -1}) = (b'_1, \ldots, b'_{d -1})$, 
$(c_1,\ldots, c_{d - 1}) = (c'_1,\ldots, c'_{d - 1})$ and 
$(\alpha_1, \ldots, \alpha_{d - 1}) = (\alpha'_1, \ldots, \alpha'_{d - 1})$.
Suppose that $(b_{2d - 1}, b_{2d}, b_{2d + 1}, b_{2d + 2}) \neq (b'_{2d - 1}, 
b'_{2d}, b'_{2d + 1}, b'_{2d + 2})$, and let $P = (\tau_1, \ldots , \tau_s)$ be 
such that $\prod_{i = 1}^{d - 1}( \tau_{b_i} + \a_{i} \tau_{c_i}) \neq 0$.
Let $j \in \{2d - 1, 2d, 2d + 1 , 2d + 2\}$ be the smallest integer such that 
$b_j \neq b'_j$, and assume w.l.o.g. that $b_j < b'_j$. Then choosing the 
entries 
of $P$ such that 
\[
\tau_{b'_{2d -1}} + \a_{b'_{2 d}} \tau_{b'_{2 d}} + \a_{b'_{2 d+1}} 
\tau_{b'_{2d+1}}
+ \a_{b'_{2 d+2}} \tau_{b'_{2d+2}} \neq 0 \textrm{ and } 
\a_{b_{j}} \tau_{b_{j}} = - \sum_{\substack{i = 2d - 1 \\ i\neq j}}^{2d+2} 
\a_{b_{i}} \tau_{b_{i}}
\]  
(here we are taking $\a_{b_{2 d-1}} := 1$) we get that 
$P \in \Supp(\varphi(g + I_X))$ and 
$P \notin \Supp(\varphi(f + I_X))$. So we assume further that 
$(b_{2d - 1}, b_{2d}, b_{2d + 1}, b_{2d + 2}) = (b'_{2d - 1}, 
b'_{2d}, b'_{2d + 1}, b'_{2d + 2})$ and suppose that 
$(\a_{2d}, \a_{2d + 1}, \a_{2d + 2}) \neq (\a'_{2d}, \a'_{2d + 1}, \a'_{2d + 
2})$. Let $j$ be the least integer among $2d, 2d+1, 2d+2$ such that 
$\a_j \neq \a'_j$. Then we may choose $P = (\tau_1, \ldots, \tau_s)$  
such that $\prod_{i = 1}^{d - 1}( \tau_{b_i} + \a_{c_i} \tau_{c_i}) \neq 0$,
and such that 
$\tau_{b_{2d -1}} + \a'_{b_{2 d}} \tau_{b_{2 d}} + \a'_{b_{2 d+1}} 
\tau_{b_{2d+1}} + \a'_{b_{2 d+2}} \tau_{b_{2d+2}} \neq 0$ 
and 
$\tau_{b_{2d -1}} + \a_{b_{2 d}} \tau_{b_{2 d}} + \a_{b_{2 d+1}} 
\tau_{b_{2d+1}} + \a_{b_{2 d+2}} \tau_{b_{2d+2}} = 0$ so that 
$P \in \Supp(\varphi(g + I_X))$ and 
$P \notin \Supp(\varphi(f + I_X))$. This shows that the  5-tuple 
described above characterizes uniquely the monic polynomials $f$ that can be 
written as in the statement of Theorem \ref{second-min-word-general}, and if 
$f$ and $g$ are distinct such polynomials, then $\varphi(f + I_X) \neq 
\varphi(g + I_X)$. We will count the number of these polynomials.
Similarly as in the proof of Theorem \ref{num-min-words}, the number of 
distinct sequences $((b_1, c_1), \ldots, (b_{d-1}, c_{d-1}))$, 
with $b_i$ and $c_i$ as in the statement of Theorem 
\ref{second-min-word-general}  for all $i = 1, \ldots, d - 1$,  is 
\[
\frac{1}{(d - 1)!} \, \prod_{k = 0}^{d - 2} \binom{s - 2 k }{2} = 
 \frac{s(s-1) . \cdots . (s - 2d + 4)(s - 2d + 3)}{(d - 1)! \; 2^{d - 1}}. 
\]
The number of possibilities for the $d-1$-tuple $(\a_1, \ldots, \a_{d -1})$ is
$(q-1)^{d - 1}$. The number of possibilities for the 4-tuple
$(b_{2d - 1}, b_{2d}, b_{2d + 1}, b_{2d + 2})$ is 
$\binom{s - 2 d + 2}{4}$ and for the triple $(\a_{b_{2 d}}, \a_{b_{2 d+1}}, 
\a_{b_{2 d + 2}})$ is $(q - 1)^3$. 
Thus the total number of monic polynomials $f$ that can be 
written as in the statement of Theorem \ref{second-min-word-general} is
\begin{equation*} 
\begin{split}
 &\frac{s(s-1) . \cdots . (s - 2d + 4)(s - 2d + 3)}{(d - 1)! \; 2^{d - 1}}. 
(q - 1)^{d - 1} \\ &. \frac{(s - 2d + 2)(s - 2d + 1 )(s - 2d)(s - 2d - 1)}{24} 
. (q - 1)^3
\end{split}
\end{equation*}
and this number, multiplied by $q -1$, is the number of next-to-minimal 
codewords.
\end{proof}

When $2d - 2 \geq s$ we have $2(s - d) + 2 \leq s$, and there is an explicitly 
described isomorphism between $\C(s - d)$ and $\C(d)$, as explained in the 
proof of \cite[Thm. 4.5]{evalcodes}. From this isomorphism, we may deduce the 
following consequence of the above theorem.

\begin{corollary} \label{num-ntm-codew2}   
Let $f \in \fqt$ be a homogeneous monomially square-free   
polynomial of degree $d$, and assume 
that 
$2 
d - 2\geq s$.  Then 
$\varphi(f + I_X)$ is a next-to-minimal weight codeword of $\C(d)$ if and only 
if $f$ 
may 
be (uniquely) written as 
\begin{equation*} 
\begin{split}
f &= \a \Biggl( \prod_{i = 1}^{s - d - 1} (t_{ b_i}  + \a_{i} 
t_{c_i}) \Biggr)
(t_{ b_{2s - 2d }} t_{ b_{2s - 2d +1 }} t_{ b_{2s - 2d + 2}}\\  &+ 
\a_{ b_{2 s - 2 d}} t_{ b_{2 s - 2 d - 1}} t_{ b_{2s - 2d +1 }} t_{ b_{2s - 2d 
+ 
2}} + \a_{ b_{2 s - 2 d+1}} t_{ b_{2 s - 2 d - 1}} t_{ b_{2 s - 2d }}  
t_{ b_{2s - 2d + 2}} \\  & 
+ \a_{ b_{2s - 2 d+2}} t_{ b_{2 s - 2 d - 1}} t_{ b_{2 s - 2d }}  
t_{ b_{2s - 2d + 1}}  )
\prod_{\stackrel{j = 1}{j \notin B_f }}^s t_j
\end{split}
\end{equation*} 
%\begin{equation*} 
%\begin{split}
%f =& \a  \left( \prod_{i = 1}^{d - 1}( t_{b_i} + \a_{i} t_{c_i}) \right)
%(t_{b_{2 d}}t_{b_{2d +1}}t_{b_{2 d+2}}\\  &+ \a_{b_{2 d}}t_{b_{2 d-1}} t_{b_{2 
%d+1}}t_{b_{2 d+2}} + \a_{b_{2 d+1}} t_{b_{2d-1}}t_{b_{2 d}}t_{b_{2 d+2}}
%+ \a_{b_{2 d+2}}t_{b_{2 d - 1}} t_{b_{2d}}t_{b_{2 d+1}}) \prod_{\stackrel{j = 
%1}{j \notin B_f}}^s t_j 
%\end{split}
%\end{equation*}
%\marginpar{renomear os ultimos bi's. Corrigir assuncao da ordem gr lex no 
%inicio da sessao 2}
with $\a, \a_{1}, \ldots, \a_{s - d - 1},   
 \a_{b_{2s - 2d }}, \a_{b_{2s - 2 d +1}}, \a_{b_{2s - 2 d + 2}} 
\in \fq^*$,  $b_1,\ldots, b_{s - d-1}$, $c_1, \ldots$, $c_{s - d-1}$, $b_{2s - 
2d 
- 1}$, $b_{2s - 2d}$, $b_{2s - 2d + 1}$ and  
$b_{2s - 2d + 2}$  
 are $2s - 2 d + 2 $ distinct elements of $\{1,\ldots, s\}$, 
$b_i < c_i$ for all $i = 1, \ldots, s - d - 1$,
$b_1 < \cdots < b_{s - d - 1}$,  $b_{2s - 2d - 1} < b_{2s - 2d} < b_{2s - 2d + 
1} < b_{2s - 2d + 2}$, and 
\[
B_f = \{b_1,\ldots, b_{s - d-1}, b_{2s - 2d - 1}, b_{2s - 2d}, b_{2s - 2d + 1}, 
b_{2s - 2d + 2}, c_1, 
\ldots, c_{s - d-1}\}.
\]

The number of next-to-minimal weight codewords of $\C(d)$, in the case where $2 
d - 2\geq s$ is 
\[
\frac{(q-1)^{s - d +3}\prod_{i = 0}^{2s - 2d + 1} (s - i) }{(s - d - 1)! \; 
2^{s - d + 2} . 3}.
\] 
\end{corollary}
\begin{proof}
We know that $\C(d)$  is monomially  equivalent to $\C(s -d)$, and from $2d - 2 
\geq s$ we get $2(s - d) + 2 \leq s$. Thus, from Theorem 
\ref{second-min-word-general} we know that the next-to-minimal weight codeword 
of
$\C(s -d)$ is obtained from the evaluation of a polynomial of the form
\begin{equation*} 
\begin{split}
f = \a  \left(\prod_{i = 1}^{s - d - 1} (t_{u_i} \right. &+ \a_{i} 
t_{v_i}) \Biggr)
(t_{u_{2s - 2d  -1}}\\  &+ \a_{u_{2 s - 2 d}} t_{u_{2 s - 2 d}} + \a_{u_{2 s - 
2 
d+1}} t_{u_{2 s - 2d + 1}}
+ \a_{u_{2s - 2 d+2}} t_{u_{2s - 2d+2}}), 
\end{split}
\end{equation*}
where with $\a, 
\a_{1}, \ldots, \a_{s - d-1}, \a_{ u_{2 s - 2 d}}, \a_{ u_{2 s - 2 d +1}}, 
\a_{ u_{2 s - 2 d + 2}}  
\in \fq^*$,  $ u_1,\ldots,  u_{s - d-1}$, $ v_1, \ldots,  v_{s - d-1}$, $ u_{2s 
- 2d 
- 1}$, $ u_{2s - 2d}$, 
$ u_{2s - 2d + 1}$, $ u_{2s - 2d + 2}$  
 are $2s - 2 d + 2 $ distinct elements of 
$\{1,\ldots, 
s\}$, $ u_i < 
v_i$ for all $i = 1, \ldots, s - d - 1$ and
$ u_1 < \cdots <  u_{s - d -1}$, $  u_{2s - 2d - 1} <  u_{2s - 2d} <  u_{2s - 
2d + 
1} < 
 u_{2s - 2d + 2}$.
From Equation \eqref{f-fc} we know that the next-to-minimal weight codewords of 
$\C(d)$ must be obtained by the evaluation of polynomials $f^c$, with $f$ as 
above. 

Let $A_{s - d - 1} := \{u_1, \ldots, u_{s - d - 1}, v_1, \ldots, v_{s - d - 1} 
\}$.
From the proof of Corollary \ref{corol.2.4} we know that 
if $j \in \{1, \ldots, s\} \setminus A_{s - d - 1}$ and $\gamma \in \fq^*$ then 
\[
\Biggl( \;   \left(\prod_{i = 1}^{s - d - 1} (t_{u_i} + \a_{i} 
t_{v_i}) \right) \gamma t_j \; \Biggr)^c = \left(\prod_{i = 1}^{s - d - 1} 
(t_{v_i} + 
\a_{i} 
t_{u_i}) \right) \gamma \prod_{\stackrel{\ell = 1}{\ell \notin A_{s - d - 
1}\cup\{ j 
\} }}^s t_\ell
\]
so
\begin{equation*} 
\begin{split}
f^c &= \a \Biggl( \prod_{i = 1}^{s - d - 1} (t_{v_i}  + \a_{i} 
t_{u_i}) \Biggr)
(t_{u_{2s - 2d }} t_{u_{2s - 2d +1 }} t_{u_{2s - 2d + 2}}\\  &+ 
\a_{u_{2 s - 2 d}} t_{u_{2 s - 2 d - 1}} t_{u_{2s - 2d +1 }} t_{u_{2s - 2d + 
2}} + \a_{u_{2 s - 2 d+1}} t_{u_{2 s - 2 d - 1}} t_{u_{2 s - 2d }}  
t_{u_{2s - 2d + 2}} \\  & 
+ \a_{u_{2s - 2 d+2}} t_{u_{2 s - 2 d - 1}} t_{u_{2 s - 2d }}  
t_{u_{2s - 2d + 1}}  )
\prod_{\stackrel{\ell = 1}{\ell \notin A_{s - d - 
1}\cup\{u_{2s - 2d - 1}, u_{2s - 2d}, u_{2s - 2d+1}, u_{2s - 2d+2}  \} }}^s 
t_\ell
\end{split}
\end{equation*} 
One may write the polynomials of the above form in a unique way, as in the 
statement of this Corollary.

The number of next-to-minimal weight codewords, in the case where $2d - 2 \geq 
s$ may 
be obtained from Theorem \ref{num-ntm-codew}, replacing $d$ by $s - d$.
\end{proof}

\begin{example} \label{last-ex}
Using a software like Magma (v.\ \cite{magma}) one may check that 
$t_1t_2t_3 + \alpha_1 t_1t_2t_4 + \alpha_2 t_1t_3t_4 + \alpha_3 
t_2t_3t_4$ is irreducible in $\fq[t_1, t_2, t_3, t_4]$ for all $\alpha_1, 
\alpha_2, \alpha_3 \in \fq^*$, and all $\fq$ such that $4 \leq q \leq 49$. 
This shows that, in the case $2d - 2 \geq s$, when one considers
the intersection of the affine torus and hypersurfaces
of degree $d$ in $\mathbb{F}_q^s$, given by a homogeneous polynomial in 
$\fqt$ whose monomials are square-free, 
the second maximal number of points may be attained by hypersurfaces which 
are not a hyperplane arrangement.
\end{example}

\noindent 
{\small
\textbf{Acknowledgments.} 

\noindent
C. Carvalho  was partially supported by Fapemig APQ-01430-24 and CNPq PQ 
308708/2023-7 \\
N. Patanker was partially supported by an IoE-IISc Postdoctoral Fellowship.}


\begin{thebibliography}{M}
%\bibitem{bruno} B.\ Buchberger, Ein Algorithmus zum Auffinden der 
%Basiselemente 
%des Restklassenringes nach einem nulldimensionalen Polynomideal. Mathematical 
%Institute, University of Innsbruck, Austria. PhD Thesis. 1965. An English 
%translation appeared in J.\ Symbolic Comput.\ 41 (2006) 475-511.
\bibitem{becker} T. Becker and V. Weispfenning,
{\it Gr\"obner Bases - A computational approach to commutative algebra},
Berlin, Germany: Springer Verlag, 1998, 2nd. pr.


\bibitem{magma}
Bosma, W.; Cannon, J.; Playoust, C. The Magma algebra system. I. The user 
language. J. Symbolic Comput. 24 (1997), 235–-265. 


\bibitem{gb-in-coding}
Carvalho, C. Gr\"obner bases methods in coding theory. Algebra for secure and 
reliable communication modeling, 73--86, Contemp. Math., 642, Amer. Math. 
Soc., 
Providence, RI, 2015.


%\bibitem{car-2024}
%Carvalho, Cícero; López, Hiram H.; Matthews, Gretchen L. Decreasing norm-trace 
%codes. Des. Codes Cryptogr. 92 (2024),  1143--1161.

\bibitem{car-neu-2018}
Carvalho, C\'{\i}cero; Neumann, Victor G.L. On the next-to-minimal weight of 
projective Reed-Muller codes. Finite Fields Appl. 50 (2018), 382--390.



\bibitem{car-nupur} Carvalho, C.; Patanker, N. Next-to-minimal weight of toric 
codes defined over 
hypersimplices. To appear in J. Algebra Appl.\\
Available at \verb|https://arxiv.org/abs/2502.07718|\ .
%\bibitem{car-2013}
%Carvalho, C. On the second Hamming weight of some Reed-Muller type codes. 
%Finite Fields Appl. 24 (2013), 88--94.

%\bibitem{gb-in-coding}
%Carvalho, C. Gr\"obner bases methods in coding theory. Algebra for secure and 
%reliable communication modeling, 73--86, Contemp. Math., 642, Amer. Math. 
%Soc., 
%Providence, RI, 2015.

%\bibitem{car-neu-2017}
%Carvalho, C; Neumann, V.G.L. On the next-to-minimal weight of affine cartesian 
%codes. Finite Fields Appl. 44 (2017), 113--134. 


\bibitem{cox}
Cox, D.; Little, J.; O'Shea, D. Ideals, Varieties, and Algorithms, 3rd ed., 
Springer, New York, 2007.


%\bibitem{fl} Fitzgerald, J.;  Lax, R.F.
%Decoding affine variety codes using Gr\"obner bases,
%Des.\ Codes and Cryptogr. \textbf{13}(2) (1998) 147--158.

%\bibitem{geil}
%Geil, O.; Høholdt, T. Footprints or generalized Bezout's theorem. IEEE Trans. 
%Inform. Theory 46 (2000), no. 2, 635--641.
%
%\bibitem{geil2}
%Geil, Olav. On the second weight of generalized Reed-Muller codes. Des. Codes 
%Cryptogr. 48 (2008), no. 3, 323--330.

\bibitem{hansen}Hansen, J.P. Toric surfaces and error-correcting codes, in: 
Coding Theory, 
 Cryptography and Related Areas, Guanajuato, 1998, 
Springer, Berlin, 2000, 132–-142.

\bibitem{evalcodes} Jaramillo, Delio; Vaz Pinto, Maria; Villarreal, Rafael H. 
Evaluation codes and their basic parameters. Des. Codes Cryptogr. 89 (2021), 
 269--300.

\bibitem{little} Little, John B. Remarks on generalized toric codes. Finite 
Fields Appl. 24 (2013), 1--14

\bibitem{jaramillo2023}
Jaramillo-Velez, Delio; López, Hiram H.; Pitones, Yuriko. Relative generalized 
Hamming weights of evaluation codes. São Paulo J. Math. Sci. 17 (2023), no. 1, 
188--207. 
%
%
%\bibitem{rolland}
%Rolland, R. The second weight of generalized Reed-Muller codes in most cases. 
%Cryptogr. Commun. 2 (2010), no. 1, 19--40.

%\bibitem{pat-singh}
%Patanker, Nupur; Singh, Sanjay Kumar. Generalized Hamming weights of toric 
%codes over hypersimplices and squarefree affine evaluation codes. Adv. Math. 
%Commun. 17 (2023),  626--643.

\bibitem{torleiv} Torleiv, K. Codes for Error Detection, Series on Coding Theory and Cryptology, 2. World Scientific Publishing Co. Pte. Ltd., Hackensack, NJ, 2007.

%\bibitem{stichtenoth} H. Stichtenoth, 
%{\it Algebraic function fields and codes},
%Universitext,
%Springer-Verlag, Berlin, 1993.

%\bibitem{magma} Bosma, W., Cannon, J., and Playoust, C. The Magma algebra 
%system. I. The user language, J. Symbolic Comput., 24 (1997), 235-265.

\end{thebibliography}
\end{document}